\documentclass[runningheads, envcountsame, a4paper]{llncs}
\usepackage{amsfonts,amssymb,amsmath,bbm}
\usepackage[T2A]{fontenc}
\usepackage[utf8]{inputenc}        

\usepackage{microtype}

\usepackage[hyphens]{url}
\usepackage{verbatim}

\usepackage{graphicx}
\def\mpfile#1#2{\includegraphics{#1#2.pdf}}

\def\uu{\mathbbm{1}}

\let\leq\leqslant
\let\geq\geqslant
\let\sm\setminus

\def\poly{\mathop{\mathrm{poly}}\nolimits}  
\def\bin{\mathop{\mathrm{bin}}\nolimits}

\let\es\varnothing

\def\ZZ{\mathbb Z}
\def\NN{\mathbb N}

\let\epsilon\varepsilon

\let\ph\varphi

\let\Ld\Lambda

\def\A{{\cal A}}

\def\D{{\cal D}}

\def\cH{{\cal H}}

\def\0{\mathsf{false}}
\def\1{\mathsf{true}}

\def\NP{{\mathbf{NP}}}
\def\PSPACE{{\mathbf{PSPACE}}}
\newcommand*\DTIME{\ensuremath{\mathbf {DTIME}}}
\def\NIM{\mathrm{NIM}}

\title{Computational Hardness of Multidimensional Subtraction Games\thanks{The study
    has been funded by the Russian Academic Excellence Project
    '5-100'. The second author was supported in part by RFBR grant 20-01-00645 and the state assignment topic no. 0063-2016-0003.
}}

\author{V. Gurvich\inst{1}\inst{4}
\and
M. Vyalyi\inst{3}\inst{1}\inst{2}}
\institute{
 National Research University Higher School of Economics\\
\and Moscow Institute of Physics and Technology\\
\and Dorodnicyn Computing Centre, FRC CSC RAS\\
\and Rutgers University\\
\email{\{vladimir.gurvich, vyalyi\}@gmail.com}
}

\def\PW{\mathcal{P}}

\usepackage{mathrsfs}

\usepackage[normalem]{ulem}
\usepackage{color}
\usepackage{xcolor}

\usepackage{enumerate}

\usepackage{marginnote}

\begin{document}

\maketitle

\begin{abstract}
We study algorithmic complexity of solving subtraction games in a~fixed dimension with a finite difference set. We prove that there exists
a~game in this class such that any algorithm solving the game runs
in exponential time. Also we prove an existence of a game in this
class such that solving the game is PSPACE-hard.

The results are based on the construction introduced by  Larsson and
W\"astlund. It  relates subtraction games and cellular automata.
\end{abstract} 
	  
\section{Introduction}\label{Intro}	       

An algorithmic complexity of solving combinatorial games is an important area
of research. 
There are  famous
games which can be solved efficiently. The most important one is
nim. The game was introduced by Bouton~\cite{Bouton02}. It can be
solved efficiently by using the
theorem on Sprague-Grundy function for 
a~disjunctive compound (or, for brevity, sum) of games
(see~\cite{BerlekampConwayGuy,Conway,GrundySmith}).
There are several generalizations of nim solved by efficient
algorithms: the Wythoff nim~\cite{Wythoff,Fraenkel84}, the Fraenkel's
game~\cite{Fraenkel82,Fraenkel84}, the nim$(a,b)$
game~\cite{BorosGurvichOudalov13}, the Moore's
nim~\cite{Moore,JenkynsMayberry80,Boros-etc19a}, the exact $(n,k)$-nim with $2k\geq
n$~\cite{Boros-etc15,Boros-etc19a}. 

There are `slow' versions for both,
Moore's and  exact  nim~\cite{GurvichHo15,GHHC}. 
In a slow version a player can take
at most one pebble from a~heap. 
In~\cite{GurvichHo15}
P-positions of exact slow $(3,2)$-nim were described. In~\cite{GHHC} the
$(4,2)$-case was solved. 


Note that for many values of parameters the exact $(n,k)$-nim is not
solved yet and the set of P-positions  looks rather
complicated. The simplest example is  the exact $(5,2)$-nim.
Slow $(5,2)$ version of exact nim reveals a~similar behavior.

So, it was conjectured that there are no efficient algorithms
solving these variants of nim. Now we have no clues how to prove this
conjecture.

Looking for hardness results in solving combinatorial games, we see
numerous examples of $\PSPACE$-complete games, 
e.g.~\cite{Schaefer78,DemaineHearn08}. 

For nim-like games, there are results on hardness of the
\emph{hypergraph nim}. 
Given a set  $[n] = \{1, \ldots, n\}$  and an arbitrary hypergraph
$\cH \subseteq 2^{[n]} \setminus \{\es\}$  on the ground   set  $[n]$,
the game {\em hypergraph nim} $\NIM_\cH$ is played as follows.
By one move a~player chooses an edge  $H \in \cH$  and reduces
(strictly) all heaps of   $H$.
Obviously, the games of standard,  exact  and  Moore's nim considered
above  are special cases of the hypergraph nim.
For a position  $x = (x_1, \ldots, x_n)$  of $\NIM_\cH$
its {\em height}  $h(x) = h_\cH(x)$ is defined as
the maximum number of successive moves that can be made from  $x$.
A~hypergraph  $\cH$  is called {\em intersecting} if  $H' \cap H'' \ne
\varnothing$
for any two edges  $H', H'' \in \cH$.
The following two statements were proven in~\cite{Boros-etc17,Boros-etc19b}.
For any intersecting hypergraph $\cH$,  its height and  SG function are equal.
Computing the height  $h_\cH(x)$ is $\NP$-complete  already
for the intersecting hypergraphs with edges of size at most  4.
Obviously, these two statements imply that,
for the above family of hypergraphs.
computing the SG function is $\NP$-complete too.


Note that all  hardness   results mentioned above were
established for games in unbounded dimension (the number of heaps is a
part of an input).

For a fixed dimension, there is a very important result of Larsson and
W\"astlund~\cite{LarssonWastlund13}. They studied a wider class of
games, so-called  \emph{vector subtraction games}. These games were
introduced by Golomb~\cite{Golomb66}. Later they were studied under
a~different name---invariant games~\cite{DucheneRigo10}. Subtraction games include all versions of nim mentioned above. 
In these games, the positions  are $d$-dimensional vectors with
nonnegative integer coordinates. The game is specified by a~set
of $d$-dimensional integer vectors (the difference set) and  a~possible
move is a~subtraction of a~vector from the difference set. 
Larsson and W\"astlund considered subtraction games of finite dimension
with a finite difference set (MSG for brevity). 

P-positions of  a 1-dimensional MSG
form a periodic structure~\cite{CGTbook}. It gives an
efficient algorithm to solve such a~game.

In higher dimensions the MSG
behave in a~very complicated way. Larsson and
W\"astlund proved in~\cite{LarssonWastlund13} that in some fixed
dimension the equivalence problem for MSG 
is undecidable.

Nevertheless, this remarkable result does not answer the major question
about efficient algorithms solving MSG.
For example, there are polynomial time
algorithms solving the membership problem for CFL but the equivalence
problem for CFL is undecidable~\cite{HoMoUl}.

In this paper we extend arguments of Larsson and W\"astlund and prove
an existence of a MSG such that 
any algorithm
solving the game runs in exponential time. For this result we need no
complexity-theoretic conjectures and derive it from the hierarchy
theorem. Also, we prove by similar arguments an existence of a~MSG
such that solving the game is $\PSPACE$-hard.  The latter
result is not an immediate corollary of the former.  It is quite
possible that a~language $L$ is recognizable only in exponential time
but $L$ is not $\PSPACE$-hard.

The rest of the paper is organized as follows. In
Section~\ref{sect:basic} we introduce all concepts used and present
the main results. In Section~\ref{sect:outline} we outline main ideas
of the proofs. The following sections contain a~more detailed
exposition of major steps of the proofs: in Section~\ref{sect:2CA->MSG}
we describe a~simulation of a binary cellular automaton by a subtraction
game; Section~\ref{sect:TM->CA} contains a
discussion of converting a Turing machine to a binary cellular
automaton; Section~\ref{sec:parallel} presents a way to launch a
Turing machine on all inputs simultaneously.
Finally, Section~\ref{sec:proofs} contains the proofs of main results.

\section{Concepts and Results}\label{sect:basic}

\subsection{Impartial Games}

An \emph{impartial} game of two players is determined by a finite set
of \emph{positions}, by the indicated \emph{initial} position and by a  set
of possible moves. Positions and possible moves form vertices and
edges of a directed graph. 
All games considered in this paper are impartial. Also, we always
assume that the graph of a game is DAG. Therefore, each play
terminates after a finite number of moves. 

Here we restrict our attention to a \emph{normal winning condition}: the
player unable to make a move loses. 

Recall the standard classification of positions of an impartial game. If
a player who moves at a position $x$ has a winning strategy in a game
starting at the position $x$, then the position is called
N-\emph{position}. Otherwise, the position is called P-\emph{position}. 
Taking in mind the relation with Sprague-Grundy function, we
assign to a~P-position the (Boolean) value~$0$ and to an N-position the
(Boolean) value~$1$. The basic relation between values of positions  is
  \begin{equation}\label{val-eq}
    p(v) = \lnot \bigwedge_{i=1}^n p(v_i) = \bigvee_{i=1}^n \lnot
    p_(v_i)
    = [p(v_1),\dots, p(v_n)], 
  \end{equation}
where the possible moves from the position $v$ are to the positions
$v_1$, $\dots$, $v_n$. 

Using Eq.~\eqref{val-eq}, it is easy to find values for all positions
of a game in time polynomial in the number of positions. We are
interested in solving games presented by a~succinct description. So, the
number of positions is typically huge and this straightforward
algorithm to solve a~game appears to be unsatisfactory.

\subsection{Subtraction Games and Modular Games}

Now we introduce a class MSG of subtraction games. 
A  game from this class is completely specified by a
finite set
$D$ of $d$-dimensional vectors (the \emph{difference set}). We assume
that coordinates of each vector $a\in D$ are integer and their sum is
positive:
\[
\sum_{i=1}^d a_i >0.
\]

A~position of the game is  a $d$-dimensional vector
$x=(x_1,\dots,x_d)$ with non-negative integer coordinates (informally,
they are the numbers of pebbles in the heaps). A~move from
the position $x$  to a~position $y$ is possible if $x-y\in
D$. If a~player is unable to make a move, then she loses.

\begin{example}
  The exact slow $(n,k)$-nim~\cite{GurvichHo15} is an $n$-dimensional
  subtraction game with the difference set consisting of all
  $(0,1)$-vectors with exactly $k$ coordinates equal~$1$.
\end{example}

Any subtraction game can be considered as a generalization of
this example.  In general case we allow to add
pebbles to the heaps. But the total number of pebbles should diminish
at each move (the positivity condition above). It guarantees that each
play of a MSG terminates after finite number of moves.

If a difference set is a part of an input, then it is easy to see that  solving of MSG  is  $\PSPACE$-hard. To
show $\PSPACE$-hardness we reduce solving of the game NODE KAYLES
to  solving a~MSG.
Recall the rules of the game NODE KAYLES. It is played on a graph
$G$. At each move a player puts a pebble on an unoccupied vertex of
the graph which is non-adjacent to any occupied vertex. The player
unable to make a move loses.
It is known that  solving NODE KAYLES is 
$\PSPACE$-complete~\cite{Schaefer78}. So, $\PSPACE$-hardness
of solving MSG is an immediate corollary of the following proposition.

\begin{proposition}\label{NodeKayles->MSG}
  Solving of NODE KAYLES is reducible to  solving of MSG. 
\end{proposition}

\begin{proof}
  Let $G=(V,E)$ be the graph of NODE KAYELS. Construct a
  $|E|$-dimensional subtraction game with the 
  difference set $A_G$ indexed by the vertices of~$G$: $D=\{a^{(v)}: v\in
  V\}$, where
  \[
  a^{(v)}_{e}= \left\{
  \begin{aligned}
    1,&&&\text{the vertex $v$ is incident to the edge $e$},\\
    0,&&&\text{otherwise.}
  \end{aligned}
  \right.
  \]
  We assume in the definition that the coordinates are indexed by the
  edges of the graph $G$.

  Take a position $\uu$ with all coordinates equal~$1$. We are going
  to prove that this position is a P-position of the MSG $A_G$ iff the
  graph $G$ is a P-position of NODE KAYLES.

  Indeed, after subtracting a vector $a^{(v)}$, coordinates indexed by
  the edges incident to~$v$ are zero. It means that after this move it
  is impossible to subtract vectors $a^{(v)}$ and $a^{(u)}$, where
  $(u,v)\in E$.

  On the other hand, if the current position is 
  \[
  \uu - \sum_{v\in X} a^{(v)}
  \]
  and there are no edges between a vertex $u$ and the vertices of the
  set $X$, then the subtraction of the vector $a^{(u)}$ is a legal
  move at this position.

  Thus, the subtraction game starting from the position $\uu$ is
  isomorphic to the game NODE KAYLES on the graph $G$.
\qed\end{proof}

In the sequel we are interested in solving of a~particular MSG (the
difference set is fixed). In other words, we are
going to determine algorithmic complexity of the language 
$\PW(D)$ consisting of binary representations of all P-positions
$(x_1,\dots,x_d)$ of the MSG with the difference set $D$.

Our main result is unconditional hardness of this problem.

\begin{theorem}\label{th:main}
  There exist a constant $d$ and a finite set $\D\subset \NN^{d}$ such that
  any algorithm 
  recognizing the language $\PW(\D)$ runs in time  $\Omega(2^{n/11})$,
  where $n$ is the input length.
\end{theorem}

Also, we  show that there are $\PSPACE$-hard languages $\PW(D)$.

\begin{theorem}\label{th:aux}
    There exist a constant $d$ and a finite set $\D\subset
    \NN^{d}$ such that the language
    $\PW(\D)$ is $\PSPACE$-hard.
\end{theorem}



In the proofs we need a generalization of MSG---so-called $k$-modular
MSG introduced in~\cite{LarssonWastlund13}. A $k$-modular
$d$-dimensional MSG is determined by $k$ finite sets $D_0, \dots,
D_{k-1}$ of vectors in $\ZZ^{d}$. The rules are similar to the rules
of MSG. But the possible moves at a position $x$ are specified by the
set $D_r$, where $r$ is the residue of $\sum_i x_i$ modulo $k$.

\subsection{Turing Machines and Cellular Automata}

A notion of a Turing machine is commonly known.
We adopt the definition of Turing machines  from Sipser's book~\cite{Sipser}.

Cellular automata are also well-known. But we prefer to provide the
definitions for them. 

Formally, a cellular automaton (CA) $C$ is  a pair
$(A,\delta)$, where $A$ is a finite set (the \emph{alphabet}), and
$\delta\colon A^{2r+1}\to A$ is the \emph{transition function}.
The number  $r$  is called
\emph{the size of a~neighborhood}.
The automaton operates on an infinite tape consisting of
\emph{cells}. Each cell carries a symbol from the alphabet. Thus,
a~\emph{configuration} of $C$ is a~function $c\colon \ZZ\to A$. 

At each step CA changes the content of the tape using the transition
function. If a~configuration before the step is $c$, then the
configuration after the step is $c'$, where
\[
c'(u) = \delta\big(c(u-r), c(u-r+1),\dots, c(u), \dots, c(u+r-1),
c(u+r)\big). 
\]
Note that changes are local: the content of a cell depends on the content
of  $2r+1$ cells in the neighborhood of the cell.

We assume that there exists a blank symbol $\Ld$ in the alphabet and
the transition function satisfies the condition 
$\delta(\Ld,\dots,\Ld)=\Ld$ (``nothing generates nothing''). This
convention guarantees that  configurations containing only
a finite number of non-blank symbols produce configurations with
 the same property.

A 2CA (a binary CA) is a CA with the binary alphabet $\{0,1\}$. Due
to relation with games, it is convenient to assume that $1$ is the
blank symbol in 2CAs.

It is well-known that Turing machines can be simulated by CA with
 $r=1$ and any CA can be simulated
by a~2CA (with  a larger size of a~neighborhood). In the proofs we need
some specific requirements on these simulations. They will be
discussed later in Section~\ref{sect:TM->CA}.

\section{Outline of the Proofs}\label{sect:outline}

Both Theorems~\ref{th:main} and~\ref{th:aux} are proved along the same
lines.
\begin{enumerate}
\item Choose a hard language $L$ and fix a Turing machine $M$
  recognizing~it. 
\item Construct another machine $U$ which simulates an operation
  of~$M$ 
  \emph{on all inputs in parallel} (a~realization of this idea is
  discussed in Section~\ref{sec:parallel}). 
\item The machine $U$ is simulated by a CA $C_U$. The cellular
  automaton $C_U$ is simulated in its turn by a 2CA $C_U^{(2)}$ (see
  Section~\ref{sect:TM->CA} for the details). And
  this $C_U^{(2)}$ is simulated by $d$-dimensional MSG $\D_U$ (see
  Section~\ref{sect:2CA->MSG}), where $d$ depends on  $C_U^{(2)}$.
\item It is important that the result of operation of $M$ on an input $w$
  is completely determined by the value of a specific position of
  $\D_U$ and this position is computed in polynomial time. So, it
  gives a~polynomial reduction of  the language $L$ to $\PW(D_U)$.
\item Now a theorem follows  from a hardness assumption on the language
  $L$.
\end{enumerate}

\section{From Cellular Automata to Subtraction Games}\label{sect:2CA->MSG}

In this section we follow the construction of Larsson and
W\"astlund~\cite{LarssonWastlund13} with minor changes.

\subsection{First Step: Simulation of a 2CA by a 2-dimensional Modular Game}

Let $C= (\{0,1\}, \delta)$ be a 2CA. The symbol $1$ is assumed to
be blank: $\delta(1,\dots,1)=1$. We are going to relate evolution of
$C$ starting from the configuration $c(0)=(\dots 11011\dots)$ with the
values $p(x)$ of positions of a 2-dimensional $2N$-modular MSG
$\D'_C$. The value of $N$ depends on $C$ and we will choose it greater than~$r$. 

The exact form of the relation is as follows. Time arrow is a
direction $(1,1)$ in the space of game positions, while the coordinate
along the  automaton tape is in the direction $(1,-1)$. 

The configuration of $C$ at moment $t$  corresponds to positions on a
line  $x_1+x_2 = 2Nt$. The cell coordinate is  $u = (x_1-x_2)/2$, as
it shown in Fig.~\ref{pic:CAgame} ($N=1$). For the configuration
$(\dots 11011\dots)$ we assume that $0$ has the coordinate~$0$ on the automaton
tape.

\begin{figure}
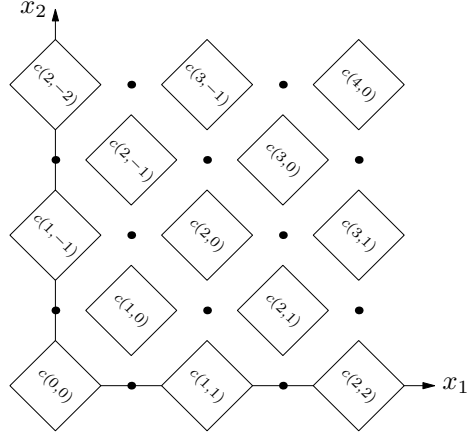

    \centering
  \mpfile{CA}{12}
  \caption{Encoding  configurations of 2CA by positions of a modular MSG}\label{pic:CAgame}
\end{figure}

The relation between the content of the automaton tape  and the values of positions
of the game $\D'_C$ is
\begin{equation}\label{2CA=mSG}
 c(t,u) = p(Nt+u,Nt-u)\quad\text{for}\ |u|\leq Nt.
\end{equation}
The choice of the initial configuration implies that if $|u|>Nt>rt$, then  $c(t,u)=1$. To extend the relation to this area, we extend
the value function $p(x_1,x_2)$ by setting $p(x_1,x_2)=1$ if either
$x_1<0$ or $x_2<0$. In other words, we introduce dummy positions with
negative values of coordinates. These positions are regarded as
terminal and having the value~$1$. 
Note
that for the game evaluation functions $[\dots]$ the equality
$
[p_1,\dots, p_k, 1,\dots,1]=[p_1,\dots,  p_k]
$
holds, 
i.e. extra arguments with the value $1$ do not affect the function
value. So, the dummy positions do not change the values of real
positions of a game.

The starting configuration $c(0)=(\dots 11011\dots)$
satisfies this relation for any game: the position $(0,0)$ is a P-position.

To maintain the relation~\eqref{2CA=mSG}, we should choose 
an appropriate modulus and
difference sets. 

Note  that the Boolean functions $[p_1,\dots, p_n]$ defined by
Eq.~\eqref{val-eq} form a complete basis: any Boolean function is
represented by a circuit with gates  $[\dots]$. It is enough
to check that the functions from the standard complete basis can be expressed
in the basis $[\dots]$:
\[
  \lnot x = [x],\quad
  x\lor y = [[x],[y]],\quad
  x\land y =[[x,y]].
\]

Now take a circuit in the basis $[\dots]$ computing the transition function
of the 2CA~$C$. The circuit is a sequence of assignments $s_1,\dots,
s_N$ of the form  
\[
  s_j:= [\text{list of arguments}],
\]
where arguments of the $j$th assignment may be  the input
variables or the values of previous assignments $s_i$, $i<j$. The
value of the last assignment $s_N$ coincides with the value of the
transition function   $\delta(u_{-r},\dots,
u_{-1},u_0, u_1,\dots, u_r)$. 

For technical reasons we require that the last assignment
$s_N$ does not contain the input variables $u_i$. It is easy to satisfy
this requirement: just start a circuit with  assignments in the form
$s_{i+r+1}=[u_i]$; $s_{i+3r+2}=[s_{i+r+1}]$, where $-r\leq i\leq r$, and substitute a~variable
$u_i$ in the following assignments by $s_{i+3r+2}$. 
The circuit size of the modified circuit is obviously
greater than~$r$.

We extend the relation~\eqref{2CA=mSG} to intermediate positions in
the following way
\begin{equation}\label{eq:intermediate}
\begin{aligned}
  &p(Nt+i,Nt-i)=c(t,i),\\
  &p(Nt+i+j,Nt-i+j)= s_j, &&1\leq j<N,
\end{aligned}
\end{equation}
where $s_j$ is the value of $j$th assignment of the circuit for the
input variables values $c(t,i-r), \dots, c(t,i), \dots,c(t,i+r)$.

\begin{proposition}\label{modDiff}
  There exist  sets $\D_j$ such that the relation~\eqref{eq:intermediate} holds for values of
  the modular game $\D'_C$ with the difference sets $\D_j$.
\end{proposition}


\begin{proof}
For each line $x_1+x_2 = 2Nt+2j$ we  specify the difference set $\D_{2j}$
according to the arguments of an assignment $s_j$. The sets with odd
indices are unimportant and may be chosen arbitrary.

If an input variable $u_k$ is an argument
of the assignment $s_j$, then we include in the set $\D_{2j}$ the vector
$(j-k, j+k)$. Since 
\[
(Nt+i+j, Nt-i+j) = (Nt+i+k, Nt-i-k) + (j-k, j+k),
\]
it guarantees that there exists a legal move from the position $(Nt+i+j,
Nt-i+j)$ to the position $(Nt+i+k, Nt-i-k)$.

If the value of an intermediate assignment $s_k$ is an argument of the
assignment $s_j$, then we include in the set $\D_{2j}$ the vector
$(j-k, j-k)$. It guarantees that there exists a move from the position
$(Nt+i+j, Nt-i+j)$ to the position $(Nt+i+k, Nt-i+k)$.

The rest of the proof is by induction on the parameter $A=2Nt+2i$,
where $t\geq0$, $0\leq i<N$. For $A=0$ we have $t=0$ and $i=0$. So the
relation~\eqref{eq:intermediate} holds as it explained above. Now
suppose that the relation holds for all lines $x_1+x_2 = A'$,
$A'<2Nt+2j$. To complete the proof, we should verify the relation on the
line $x_1+x_2 = 2Nt+2j$. From the construction of the sets $\D_{2j}$
and  the induction hypothesis 
we conclude that
\[
p(Nt+i+j,Nt-i+j) = [\text{arguments of the assignment $s_j$}].
\]
Here arguments of  the  assignment
 $s_j$ are  the values of the input variables and the values of  previous assignments
in the circuit computing the
transition function $\delta(c(t,u-r), \dots, c(t,u),\dots,
c(t,u+r))$. 

The last touch is to note that the value of the $N$th assignment is
 just the value 
$c(t+1, u) = \delta(c(t,u-r), \dots, c(t,u),\dots,
c(t,u+r))$.
\qed
\end{proof}

Note that  the game $\D'_C$ has the
property: if there is a legal move from $(x_1,x_2)$ to $(y_1,y_2)$,
then either $x_1+x_2\equiv 0\pmod{2N}$ or the residue of $(y_1,y_2)$
modulo $2N$ is less than the residue of $(x_1,x_2)$ (we assume the
standard representatives for residues: $0,1,\dots, 2N-1$). Also,
$x_1+x_2\not\equiv y_1+y_2\pmod{2N}$ since the input variables are not
arguments of the final assignment.


\subsection{Second Step: Simulation of a 2CA by a $(2N+2)$-dimensional Subtraction Game}

To exclude modular conditions we use the trick suggested
in~\cite{LarssonWastlund13}.







Using the 2-dimensional modular game $\D'_C$ constructed above
we construct a $(2N+2)$-dimensional MSG $\D_C$ with the difference
set 
\[
\D = \big\{(a_1, a_2,0^{2N}
\big)+e^{(j)}-e^{(k)}: 
(a_1,a_2)\in D_j,\ k=j- a_1-a_2\pmod{2N}\big\}.
\]
Here $e^{(i)}$
is the $(i+2)$th coordinate vector: $e^{(i)}_{i+2} =
1$, 
$e^{(i)}_s=0$ for $s\ne i+2$.

\begin{proposition}\label{mSG->SG}
  The value of a~position $(x_1,x_2, 0^{2N})+ e^{(2r)}$ of the game $\D_C$ equals
  the value of a~position $(x_1,x_2)$ of the modular game $\D'_C$ if $2r\equiv x_1+x_2\pmod {2N}$.
\end{proposition}
\begin{proof}
  Induction on $t=x_1+x_2$. The base case $t=0$ is  due
  to the convention on the values of dummy positions (with negative
  coordinates). 

  The induction step. A~legal move at a~position
  $(Nt+i+j,Nt-i+j,0^{2N})+e^{(2j)}$ is to a position 
  $(Nt+i+j,Nt-i+j,0^{2N})-(a_1,a_2,0^{2N})+e^{(2s)}$, where $2s\equiv
  2j-a_1-a_2\pmod {2N}$ and $(a_1,a_2)\in \D_{2j}$. It corresponds to
  a move from $(Nt+i+j,Nt-i+j)$ to $(Nt+i+j-a_1,Nt-i+j-a_2)$ in the
  modular game.
\qed\end{proof}

From Propositions~\ref{modDiff} and~\ref{mSG->SG} we conclude

\begin{corollary}\label{CA->SG}
  For any 2CA $C$ there exist an integer $N$ and a $(2+2N)$-dimensional MSG $\D_C$ such that
  the relation
\[
 c(t,u) = p(Nt+u,Nt-u,0,0,\dots,0,1)\quad\text{holds for}\ |u|\leq Nt.
\]
\end{corollary}


\section{From Turing Machines to Cellular Automata}\label{sect:TM->CA}

In this section we outline a~way to simulate a Turing machine by a
binary cellular automaton. It is a standard simulation, but we will put
specific requirements.

Let $M = (Q,\{0,1\},\Gamma, \Ld, \delta_M, 1, 2, 3)$ be a Turing
machine, where $Q = \{1,\dots,q\}$, $q\geq3$, is the set of states,
the input alphabet is binary, $\Gamma=\{0,1,\dots,\ell\}$ is the tape
alphabet, $\ell>1$ is the blank symbol, $\delta_M\colon Q\times\Gamma\to
Q\times \Gamma\times\{+1,-1\} $ is the transition function, and $1, 2,
3$ are the initial state, the accept state, the reject state
respectively.

We encode a configuration of $M$ by a doubly infinite string $c\colon
\ZZ\to A$, where  $A = \{0,\dots,q\}\times \{0,\dots,\ell\}$,
indicating the head position by a~pair $(q,a)$, $q>0$, $a\in
\Gamma$; the content
of any other cell is encoded as  $(0,a)$, $a\in \Gamma$.

Let $c_0, \dots, c_t,\dots$ be a sequence of encoded configurations produced
by  $M$ from the starting configuration $c_0$.
It is easy to see that $c_{t+1}(u)$ is determined by $c_t(u-1)$,
$c_t(u)$, $c_{t}(u+1)$. In this way we obtain  the CA $C_M= (A,\delta_C) $ over the
alphabet $A$ with  the transition function
$\delta_C\colon A^3\to A$ simulating operation of~$M$ in encoded
configurations. It is easy to see that $\Ld = (0,\ell)$ is the blank symbol: $\delta_C(\Ld,\Ld,\Ld) = \Ld$.

The next step is to simulate $C_M$ by a 2CA $C_M^{(2)}$. For this
purpose we use an automaton $C'_M= (A',\delta'_C) $ isomorphic to
$C_M$, where $A' = \{0,\dots,L-1\}$ and $L = (|Q|+1)\cdot|\Gamma|$.
The transition function $\delta'_C$ is defined as follows
\[
\delta'_C (i,j,k) = \pi (\delta_C(\pi^{-1}(i), \pi^{-1}(j), \pi^{-1}(k))), 
\]
where $\pi \colon A\to A'$ is a bijection. To keep a~relation between
the starting
configurations 
we require that $\pi(\Ld) = 0$, $\pi((1,\ell)) = 1$.
Recall that $1$ is the initial state of $M$ and
$\ell$ is the blank symbol of $M$.

To construct the transition function of $C_M^{(2)}$ we encode symbols
of $A'$ by binary words of length $L+2$ as
follows 
\[
\ph(a) = 1^{1+L-a}0^a1.
\]
In particular, $\ph(0) = \ph(\pi (\Ld)) = 1^{L+2}$ and 
$\ph(1) = \ph(\pi(1,\ell)) = 1^L01$. The encoding $\ph$ is naturally
extended to words in the alphabet $A'$ (finite or infinite). 

Thus the starting configuration of $M$ with the
empty tape corresponds to the configuration $\dots1110111\dots$ of
$C_M^{(2)}$. Recall that $1$ is the blank symbol of $C_M^{(2)}$.

With an
abuse in notation, we denote
below by  $\ph$ the extended encoding of configurations in the
alphabet $A'$ by doubly infinite binary words.
We align configurations in the following way:
if $i = q(L+2)+k$, $0\leq k<L+2$, then $\ph(c)(i)$ is a
$k$th bit of the $\ph(c(q))$. 

The size of a neighborhood of $C_M^{(2)}$ is $r = 2(L+2)$. To define
the transition function $\delta^{(2)}_M$ we use a~local inversion
property of the encoding~$\ph$:
 looking at the $r$-neighborhood of an
$i$th bit of $\ph(c)$, where $i = q(L+2)+k$, $0\leq k<L+2$, one can restore symbols $c(q-1)$, $c(q)$,
$c(q+1)$ and the position $k$ of the bit provided the neighborhood
contains zeroes ($0$ is the non-blank symbol of $C_M^{(2)}$).  Note
that if the neighborhood of a bit does not contain zeroes, then the
bit is a part of encoding of the blank symbol $0$ of $C'_M$ and,
moreover,  $c(q-1)=c(q)=c(q+1)=0$.

\begin{lemma}\label{restore}
  There exists a function $\delta^{(2)}_C\colon \{0,1\}^{2r+1}\to\{0,1\}$ such
    that a 2CA $C_M^{(2)}=(\{0,1\},\delta^{(2)}_C)$
simulates $C'_M$:  starting from  $b_0 = \dots1110111\dots$, it produces  the sequence of 
configurations $b_0,b_1,\dots$ such that $b_t = \ph(c_t)$ for any $t$,
where $(c_t)$ is the sequence of configurations produced by $C'_M$
starting from the configuration $c_0=\dots0001000\dots$
\end{lemma}

\begin{proof}
The function $\delta^{(2)}_C$ should satisfy the following property. If
$b=\ph(c)$, then 
\begin{equation}\label{delta2-def}
\delta^{(2)}_C\big((b(i-r),\dots, b(i), \dots , b(i+r)\big)
= 
\ph\big(\delta'_C\big(c(q-1), c(q), c(q+1))\big)(k)
\end{equation}
for all integer $i =  q(L+2)+k$, $0\leq k<L+2$. This property means
that applying the function $\delta^{(2)}_C$ to $b$ produces the
configuration $b_1=\ph(c_1)$, where $c_1$ is the configuration
produced by the transition function $\delta'_C$ from the
configuration~$c$. Therefore the sequence of configurations produced
by $C_M^{(2)}$ starting at $\ph(c_0)$ is the sequence of the encodings of
configurations  $c_t$ produced by $C'_M$ starting at~$c_0$.

Note that $\ph(c(q-1))$, $\ph(c(q))$ and $\ph(c(q+1))$ are in the
$r$-neighborhood of a~bit $i$. 

Thus, from the condition on blank symbols in the alphabets $A'$ and
$\{0,1\}$, we conclude the required property holds if the
$r$-neighborhood of a~bit $i$  does not contain zeroes (non-blank
   symbols of $C_M^{(2)}$). 
In this case the $i$th bit of $b$ is a~part of encoding of the blank
symbol $0$ in the alphabet $A'$ and,
moreover,  $c(q-1)=c(q)=c(q+1)=0$.

Now suppose that the $r$-neighborhood of the $i$th bit contains
zeroes. Take the nearest zero to this bit (either from the left or
from the right) and the maximal series $0^a$ containing it. The series
is a part of the encoding of a symbol in $c$. So, there are at least
$1+L-a$ ones to the left of it. They all should be in the
$r$-neighborhood of the bit.  Thus we locate an encoding of a symbol
$c(q+q')$, $q'\in\{-1,0,+1\}$, and we are able to determine $q'$
(depends on relative position of the $i$th bit with respect to
the first bit of the symbol located). So, the  symbols
$c(q-1)$, $c(q)$, $c(q+1)$ can be restored from the
$r$-neighborhood of the $i$th bit.
Moreover,  a~relative position $k$ of the $i$th bit in  $\ph(c(q+q'))$ can
also be restored. 

Because the symbols $c(q-1)$, $c(q)$, $c(q+1)$ and the position $k$
are the functions of the  $r$-neighborhood of the bit $i$, it is
correct to define the function  $\delta^{(2)}_C$ as
\[
\delta^{(2)}_C\big(u_{-r},\dots,u_0,\dots,u_r\big) = 
\ph\big(\delta'_M(c(q-1), c(q), c(q+1))\big)(k)
\]
if the restore process is successful on $(u_{-r},\dots,u_0,\dots,u_r)$; 
for other arguments, the function can be defined arbitrary. 
It is clear that this function satisfies the property~\eqref{delta2-def}.
\qed\end{proof}

\section{A Parallel Execution of a Turing Machine}\label{sec:parallel}

The last construction needed in the main proofs is
a Turing machine $U$ simulating an operation of a Turing machine $M$
\emph{on all inputs}. The idea of simulation is well-known. But,
again, we need to specify some details of the construction.

We assume that on each input of length $n$ the machine $M$ makes at most
$ T(n)>n$ steps.

The alphabet of $U$ includes the set $A = \{0,\dots,q\}\times
\{0,\dots,\ell\} $ (we use notation from the previous section) and
additional symbols. 

The machine $U$ operates in \emph{stages} while its tape is divided
into \emph{zones}. The zones are surrounded by the delimiters,
say, $\triangleleft$ and $\triangleright$. We assume that
$\triangleleft$ is placed to the  cell~$0$.
Also the zones are
separated by a delimiter, say, $\diamond$.  An operation of $M$ on a
particular input $w$ is simulated inside a separate zone.

Each zone consists of three blocks.
as pictured in Fig.~\ref{TMzone}. 
\begin{figure}[!h]
  \centering
  \mpfile{TM}{1}
  \caption{A zone on the tape of $U$}\label{TMzone}
\end{figure}

The first  block of a zone has the size $1$. It carries $(0,1)$ iff $M$
accepts the input written in the second block. Otherwise it carries
$(0,0)$.  The last block contains a~configuration of $M$ represented
by a word over the alphabet $A$ as described in Section~\ref{sect:TM->CA}. Blocks in a zone are
separated by a delimiter, say \#. 

At start of  a stage   $k$ there are  $k-1$ zones corresponding to
the inputs  $w_1$, $w_2$, $\dots$, $w_{k-1}$ of $M$. 
We order binary
words by their lengths and words of equal length are ordered
lexicographically. The last block of a zone~$i$ contains the configuration
of 
$M$ after running $k-1-i$ steps on the input  $w_i$. 

During the stage $k$, the machine  $U$ moves along the tape from
$\triangleleft$ to $\triangleright$ and
in each zone simulates the next step of operation   of $M$. At the end
of the stage the machine $U$ writes a~fresh zone with the input $w_k$ and the
initial configuration of $M$ on this input. The initial configuration is
extended in both directions by white space of size  $T(n)$, as it shown
in Fig.~\ref{TMinit}. 

\begin{figure}[!h]
  \centering
  \mpfile{TM}{2}
  \caption{A fresh zone on the~stage $k$}\label{TMinit}
\end{figure}

When an operation of $M$ on an input $w_k$ is finished, the machine
$U$ updates the result block and does not change the zone on
subsequent stages.

In reductions below we need $U$  satisfying specific properties.

\begin{proposition}\label{goodU}
  If $T(n)= C 2^{n^k}$ for some integer constants $C\geq1$, $k\geq1$, then there exists
  $U$ operating as it described above such that 
\begin{enumerate}
\item \label{time-bnd}
  $U$ produces the result of operation of $M$ on input $w$ in time
  $<2^{4n^k}$, where $ n = |w|$.
\item \label{mod-cond}
  The head $U$ visits the first blocks of zones only on steps~$t$
  that are divisible by~$3$. 
\end{enumerate}
\end{proposition}


\begin{proof}
  Recall that  operation of the machine $U$ is
  divided in stages. 

  During the stage $k$, the machine  $U$ moves along the tape from the
  left to the right and
  in each zone simulates the next step of operation   of $M$. At the end
  of the stage the machine $U$ writes a~fresh zone with the input $w_k$ and the
  initial configuration of $M$ on this input. The configuration is
  extended in both directions by white space of size  $T(n)$.

  At first, we show how to construct   a machine $U'$ satisfying the
  property~\ref{time-bnd}. More exactly, we explain how to construct a
  machine satisfying the following claims.

  {\bfseries\itshape Claim 1.} Updating a configuration of the simulated machine $M$ into a zone
  takes a time $O(S)$, where $S$ is the size of the zone.

    A straightforward way to implement the update is the following. The
    head of $U'$ scans the zone until it detects a symbol $(q,a)$ with
    $q>0$. It means that the head of the simulated machine~$M$ is over the
    current cell. Then $U'$ updates the neighborhood of the cell
    detected with respect to the transition function of~$M$. After
    that $U'$ continues a motion until it detects the next zone.

    If a machine $M$ finishes its operation on a~configuration written
    in the current zone, then additional actions should
    be done. The machine $U'$ should update the result block. For this
    purpose it returns to the left end of the zone, updates the result
    block and continues a motion to the right until it detects the next zone.

    So, each cell in the zone is scanned $O(1)$ times. The total time
    for update is $O(S)$.

  {\bfseries\itshape Claim 2.}   A fresh zone on the stage $k$ is
  created in time  $O(n^k T(n))$, where  $n= |w_k|$.

  Creation of the result block takes a time  $O(1)$.

  To compute the next input word the machine $U'$ copies the previous
  input into the second block of the fresh zone. The distance between
  positions of the second blocks is
  $4+|w_{k-1}|+2T(|w_{k-1}|)=O(T(n))$. Here we
  count three delimiters occuring between the blocks and use the
  assumption that $T(n)>n$. The machine $U'$ should copy at most $n$
  symbols. So, the copying takes a time $O(nT(n))$. 

  After that, the machine $U'$ computes the next word in the
  lexicographical order. It can be done by adding $1$ modulo
  $2^{|w_{k-1}|}$ to $\bin(w_{k-1})$, where $\bin(w) $ is the integer
  represented in binary by $w$ (the empty word represents~0).  It
  requires a time $O(n)$. If an overflow occurs, then the machine
  should write an additional zero. It also requires a time $O(n)$.

  To mark the third block in the fresh zone the machine $U'$ computes
  a binary representation of $T(n)$
  by a polynomial time algorithm using the second block as an input to
  the algorithm (thus, $n$ is given in unary). Then it makes $T(n)$
  steps to the right using the computed value as a counter and decreasing
  the counter each step. The counter should be moved along a tape to
  save a time. 
  The length of binary representation of $T(n)$
  is $O(n^k)$. So, each step requires $O(n^k )$ time and totally
  marking of $T(n)$ free space  requires $O(n^k T(n))$ time. 

  Then $U'$ copies the input word $w_k$ to the right of marked free
  space. It requires $O(n T(n))$ time. The first cell of the copied
  word should be  modified to indicate  the
  initial state of the simulated machine~$M$.  And, finally, it
  repeat the marking procedure to the right of the input. 

  The overall time is $O(n^k T(n))$.

  Let us prove the property~\ref{time-bnd} is satisfied by the
  machine~$U'$. Counting time in stages, the zone corresponding to an
  input word $w$ of length $n$ appears after $\leq 2^{n+1}$
  stages. After that the result of operation of $M$ appears after
  $\leq T(n)$ stages.

  Let $s=|w_k|$. At stage $k$ there are at most $2^{s+1}$
  zones. Updating the existing zones requires time $O(2^{s}(s+T(s)))$ due
  to Claim~1. Creation of a fresh zone requires time $O(s^k T(s))$ due
  to Claim 2. Thus, the overall time for a stage is 
  \[
  O\big(2^{s+1}(s+T(s)) +s^k T(s) \big) = O(2^sT(s)).
  \]

  Therefore, the result of operation of $M$ appears in time
  \[
  O\big((2^{n+1}+T(n)) T(n)^2\big) = O(2^{3n^k})< 2^{4n^k}
  \]
  for sufficiently large~$n$.

  Now we explain how to modify the machine $U'$ to satisfy the
  property~\ref{mod-cond}. Note that the result block of a zone is
  surrounded by delimiters: \# to the right of it and either
  $\triangleleft$ or $\diamond$ to the left. 

  We enlarge the state set of $U'$ adding a counter modulo~3. It is
  increased by $+1$ each step of operation. If the head of a modified machine
  $U$ is over the $\triangleleft$ or $\diamond$ and $U'$ should go  to the
  right, then the machine $U$ makes dummy moves in the opposite
  direction and back to ensure that it visits the cell to the right on
  a~step $t$ divisible by 3. In a similar way the machine $U$
  simulates the move to the left from the cell carrying the delimiter \#.
\qed\end{proof}

\section {Proofs of the Main Theorems}\label{sec:proofs}

\begin{proof}[of Theorem~\ref{th:main}]
  Time hierarchy theorem~\cite{Sipser} implies that
  $\DTIME(2^{n/2})\subset \DTIME(2^{ n})$. Take a language $L\in
  \DTIME(2^{ n})\sm\DTIME(2^{n/2})$. For some constant  $C$ there
  exists a Turing machine  $M$ recognizing  $L$ such that $M$  makes at most
  $ T(n)=C2^n$ steps on inputs of length~$n$. 

  Apply the construction
  from Section~\ref{sec:parallel} and Proposition~\ref{goodU} to
  construct  the machine~$U$. Then convert $U$ into 2CA $C^{(2)}_U$ as
  it described in Section~\ref{sect:TM->CA}. We put an additional
  requirement on the bijection $\pi$, namely, $\pi(0,(0,1))=L-1$. 
  It locates the result of computation of $M$  in the  third
  bit of the encoding of the result block.

  Finally, construct
  $O(1)$-dimensional MSG $\D_C$ as it described in
  Section~\ref{sect:2CA->MSG}. The dimension $2N+2$ of the game is determined
  by the machine~$M$.

  Due to Corollary~\ref{CA->SG} the symbol $c(t,u)$ on the tape of
  $C^{(2)}_U$  equals the value of position
  $(Nt+u,Nt-u,0,0,\dots,0,1)$ of the game. 

  Suppose that we have an algorithm $\A$ to solve the game $\D_C$ in
  time $T_\A(m)$.

  Consider the following  algorithm recognizing $L$. 

  On an~input $w$ of
  length~$n$ do:
  \begin{enumerate}
  \item\label{first-step} Compute the number $k$ of the~zone
  corresponding to the input~$w$ of length $n$.
  \item Compute the position $u$ of the bit carrying the result of
  computation of $M$ on input~$w$ in the image of the result block of
  the zone~$k$.
   \item\label{last-step} Set $t=2^{4n}$. 
   \item Apply the algorithm $\A$ to compute the value of the position
   $$(Nt+u,Nt-u,0,0,\dots,0,1)$$ of the game and return the result. 
  \end{enumerate}

  Correctness of the algorithm is ensured by previous constructions
  and by the property~\ref{mod-cond} of Proposition~\ref{goodU}.
  The latter guarantees that at moment $t=2^{4n}$
  the head of $U$ is not on the result block. Thus the third bit of
  the encoding of the block is $1$ iff $M$ accepts $w$.

  It can be easily verified (see
  Proposition~\ref{calculations} below) that
  the first two steps of the algorithm can be done in time $\poly (n)$  and
  $u=O(2^{3n})$. The property~\ref{time-bnd} of
  Proposition~\ref{goodU} ensures that $U$ produces the result of $M$
  on the input $w$ in time $<2^{4n}$ for sufficiently large~$n$.

  Thus, the total running time of the algorithm is at most $\poly
  (n)+T_\A(5n) $.  But by choice of $L$ it is $\Omega (2^{n/2})$. We
  conclude that $T_\A(m)=\Omega (2^{m/11})$.
\qed\end{proof}

To complete the proof of Theorem~\ref{th:main}, we provide the proof
of a technical claim made.

\begin{proposition}\label{calculations}
  The first two steps of the algorithm in the proof of
  Theorem~\ref{th:main} can be done in time $\poly (n)$ for
  $T(n)=C2^{n^k}$ and $u=O(2^{3n})$ if $T(n)= C2^n$.
\end{proposition}

\begin{proof}
  For the first step, note that $k= 2^{n}+\bin(w)$. Indeed, there are $2^{n}-1$
shorter words, all of them precede $w$ in the ordering of binary words
we use. Also there are exactly $\bin(w)$ words of length $n$ preceding the
word~$w$. The formula for $k$ follows from these observations (note
that we count words starting from~1).

It is quite obvious now that $k$ is computed in polynomial time.

For the second step, we should count the sizes of zones preceding the
zone for $w$ and add a~constant to take into account delimiters. 
Let count the size of a zone including the delimiter to the left of
it. Then the size of a~zone for an input word of length $\ell$ is
\[
1+1+1+\ell+1+\ell+2T(\ell) = 4+2\ell+2T(\ell).
\]
There are $2^\ell$ words of length $\ell$. Thus, 
the total size of the zones preceding the zone of $w$ is 
\[
S = \sum_{\ell=0}^{n-1} 2^\ell (4+2\ell+2T(\ell)) + \bin(w) (4+2n+2T(n)) + 2
\]
For $T(n) = C2^{n^k}$ this expression can be computed in polynomial
time in $n$ by a straightforward procedure (the expression above has
$\poly(n)$ arithmetic operations and the results of these operations
are  integers represented in binary
by $\poly(n)$ bits). 

Thus, the result block of the zone of $w$ is $S+1$ (the delimiter to
the left of the zone adds~1). 

To compute $u$ we should multiply $S+1$ by $L=O(1) $ (the size of
encoding) and add~3 (because the third bit indicates the result of
computation of the simulated machine~$M$).

All these calculations can be done in polynomial time. 
If $T(n)=C2^n$, then we upperbound $u$ as follows
\[
u\leq L\big(n 2^n (4+2n+2C2^n) + 2^n (4+2n+2C2^n) + 3\big)+3 =
O(2^{3n}).
\]
\qed\end{proof}

\begin{proof}[of Theorem~\ref{th:aux}]
  Take a $\PSPACE$-complete language $L$ and repeat arguments from the
  previous proof using an upper bound $T(n)=C 2^{n^k}$  of the
  running time   of a machine $M$ recognizing $L$. The bound
  follows from the standard counting of the number of configurations in an
  accepting computation using polynomial space. 

  At the step~\ref{last-step} set $t =2^{4n^k}$. It gives a~polynomial
  reduction of $L$ to $\PW(\D_C)$: $w\mapsto(Nt+u,Nt-u,0,0,\dots,0,1)
  $.
\qed\end{proof}

\bibliographystyle{splncs03}
\bibliography{games}
\end{document}